\documentclass[12pt]{article}
\usepackage{amsmath}
\usepackage{latexsym}
\usepackage{amsfonts}
\usepackage{amssymb}
\usepackage{amsthm}
\usepackage{amstext}
\usepackage{bbm,dsfont}
\usepackage{graphicx}
\usepackage[normalem]{ulem}
\usepackage{caption}
\usepackage{enumerate}

\usepackage{color, palatino}
\usepackage[top=3cm, bottom=3cm, left=2cm, right=2cm]{geometry}

\newtheorem{proposition}{Proposition}

\theoremstyle{definition}

\newtheorem{definition}{Definition}

\newcommand{\de}{\,{\rm d}}

\newcommand{\R}{\mathbb{R}} 
\newcommand{\C}{\mathbb{C}} 
\newcommand{\rational}{\mathbb Q} 

\newcommand{\nat}{\mathbb N} 

\newcommand{\hi}{\mathcal{H}} 
\newcommand{\hh}{\mathcal{H}} 
\newcommand{\ww}{\mathcal{W}} 
\renewcommand{\aa}{\mathcal{A}} 
\newcommand{\lh}{\mathcal{L(H)}} 

\newcommand{\trh}{\mathcal{T(H)}} 
\newcommand{\ip}[2]{\left\langle\,#1\,|\,#2\,\right\rangle} 
\newcommand{\kb}[2]{|#1\rangle\langle#2|} 
\newcommand{\tr}[1]{\mathrm{tr}\left[#1\right]} 

\newcommand{\id}{\mathbbm{1}} 

\newcommand{\Eoo}{\mathsf{E}}
\newcommand{\Po}{\mathsf{P}}
\newcommand{\Qo}{\mathsf{Q}}

\newcommand{\e}{{\rm e}}

\newcommand{\h}[1]{\mathcal{#1}}
\newcommand{\be}{\begin{equation}}
\newcommand{\ee}{\end{equation}}

\begin{document}\setlength{\arraycolsep}{2pt}

\title{Nonuniqueness of phase retrieval for three fractional Fourier transforms}

\author{Claudio Carmeli\thanks{DIME,  Universit\`a di Genova, Via Magliotto 2, I-17100 Savona, Italy. E-mail: {\em claudio.carmeli@gmail.com}} \and Teiko Heinosaari\thanks{Turku Centre for Quantum Physics, Department of Physics and Astronomy,  University of Turku, FI-20014 Turku, Finland. E-mail: {\em teiko.heinosaari@utu.fi}} \and Jussi Schultz\thanks{Dipartimento di Matematica, Politecnico di Milano, Piazza Leonardo da Vinci 32, I-20133 Milano, Italy \& Turku Centre for Quantum Physics, Department of Physics and Astronomy,  University of Turku, FI-20014 Turku, Finland. E-mail: {\em jussi.schultz@gmail.com}} \and Alessandro Toigo\thanks{Dipartimento di Matematica, Politecnico di Milano, Piazza Leonardo da Vinci 32, I-20133 Milano, Italy, \& I.N.F.N., Sezione di Milano, Via Celoria 16, I-20133 Milano, Italy. E-mail: {\em alessandro.toigo@polimi.it}}}

\date{\empty}

\maketitle

\begin{abstract}
We prove that, regardless of the choice of the angles $\theta_1,\theta_2,\theta_3$, three fractional Fourier transforms $F_{\theta_1}$, $F_{\theta_2}$ and $F_{\theta_3}$ do not solve the phase retrieval problem. 
That is, there do not exist three angles  $\theta_1$, $\theta_2$, $\theta_3$ such that any signal  $\psi\in L^2(\R)$ could be determined up to a constant phase by  knowing only the three intensities $|F_{\theta_1}\psi|^2$, $|F_{\theta_2}\psi|^2$ and $|F_{\theta_3}\psi|^2$. This provides a negative argument against a recent speculation by P. Jaming, who stated that three suitably chosen fractional Fourier transforms are good candidates for phase retrieval in infinite dimension. 
We recast the question in the language of quantum mechanics, where our result shows that any fixed triple of rotated  quadrature observables $\Qo_{\theta_1}$, $\Qo_{\theta_2}$ and $\Qo_{\theta_3}$ is not enough to determine all unknown pure quantum states.
The sufficiency of four rotated quadrature observables, or equivalently fractional Fourier transforms, remains an open question.
\end{abstract}

\section{Introduction}

The problem of phase retrieval deals with reconstructing the phase of a complex signal from intensity measurements, that is, measurements which give as an outcome only the modulus of the signal. This problem is encountered in a wide variety of practical circumstances such as microscopy and crystallography \cite{Stark87}. In the context of quantum mechanics, it can be traced back to W.~Pauli, who noted in a footnote in \cite{Pauli33} that the question whether or not the position distribution $\vert \psi\vert^2$ and the momentum distribution $\vert\widehat{\psi}\vert^2$ uniquely determine the wave function $\psi$, ``has still not been investigated in all its generality''. 
It was soon realized that the distributions $\vert \psi\vert^2$ and $\vert\widehat{\psi}\vert^2$ do not determine the wave function up to a phase \cite{PFQM44}, and therefore one is lead to search for a larger class of measurements that would be sufficient for the task at hand (see e.g.~\cite{Vogt78, Corbett06} for some more recent developments).

A natural direction for extending the original Pauli problem is given by the fractional Fourier transforms. The fractional Fourier transforms are a family of unitary operators $F_\theta$, $\theta\in[0,2\pi)$, on $L^2(\R)$ which generalize the usual Fourier transform in such a way that (i) $F_0=I$, (ii) $F_{\pi/2} =F$, the usual Fourier transform, (iii) $F_\pi=\Pi$, the parity operator, (iv) $F_{3\pi/2}=F^{-1}$ and (v) $F_{\theta_1} F_{\theta_2} = F_{\theta_1+\theta_2}$ where addition is understood {\em modulo} $2\pi$ (see \cite{Namias1980,McBride87}). The problem then is to determine if the knowledge of the intensities $\vert F_\theta\psi\vert^2 $ for some suitable set of angles is sufficient for determining an arbitrary signal $\psi\in L^2(\R)$ up to a constant phase. 
In the recent article \cite{Jaming14}, Jaming carried out this line of approach and obtained several interesting results in a number of cases where the signal is known to belong to some restricted class. In particular, in \cite[Theorem 5.5]{Jaming14} he proved that if $\psi$ belongs to the dense subset of finite linear combinations of Hermite functions, then already two suitably chosen fractional Fourier transforms are sufficient. This immediately raises the question if three angles would be enough for an arbitrary signal. Indeed, Jaming himself suggests that ``the fractional Fourier transform is a good candidate'' for providing three unitary operators which would guarantee the uniqueness of  phase retrieval.

It is the purpose of this Letter to show that regardless of the choice of the angles, three fractional Fourier transforms are not enough to ensure the uniqueness in the phase retrieval problem for arbitrary signals. 
We prove this by first formulating the question in the context of quantum mechanics, in which case knowledge of the modulus of the fractional Fourier transform corresponds to a measurement of a rotated quadrature observable (see formula \eqref{eq:quad-->Four} below). 
The problem is then turned into the analysis of the operator systems generated by sets of quadrature observables, rather than directly dealing with the corresponding Fourier operators. 
In this way, its solution is much simpler, as it
essentially boils down to the analysis of symplectic $2\times 2$-matrices. 
We then close this Letter by showing that our method no longer works for four angles, and therefore the exhaustive answer to the question regarding the minimal number of fractional Fourier transforms for unique phase retrieval remains an open question. These results should be compared to similar ones in the finite-dimensional setting, where it is known that uniqueness for the phase retrieval can be achieved with four unitary operators \cite{Jaming14,MoVo13}, and at least for sufficiently high dimensions this is the minimal number \cite{Moroz84,MoPe94,HeMaWo13}.

\section{Quantum mechanical formulation of the problem}
In quantum mechanics, the description of a physical system is based on a complex separable Hilbert space $\hi$. We use the notation $\ip{\cdot}{\cdot}$  for the inner product on $\hi$ which, following the convention of the physics literature, we assume to be linear in the second argument. 

Let $\lh$ and $\trh$  denote the Banach spaces of bounded and trace class operators on $\hi$, respectively. The physical {\em states} of the system are represented by elements $\varrho\in\trh$ satisfying positivity $\varrho \geq 0$ and normalization $\tr{\varrho} =1$. The states form a convex set whose extreme points, called the {\em pure states}, are precisely the one dimensional  projections $\vert \varphi \rangle\langle \varphi\vert:\hi\to\hi$, $\Vert\varphi \Vert=1$, defined via $\vert \varphi \rangle\langle \varphi\vert\psi = \langle\varphi\vert \psi\rangle \varphi$, with $\varphi$ and $\psi$ in $\hi$. 
The {\em  observables} are represented by normalized positive operator valued measures (POVMs) $\Eoo:\h B(\R)\to \lh$ where $\h B(\R) $ denotes the Borel $\sigma$-algebra of $\R$ \cite{OQP97,PSAQT82}. More precisely, an observable is a map $\Eoo:\h B(\R)\to\lh$ which satisfies (i) positivity $\Eoo(X)\geq 0$ for all $X\in \h B(\R)$, (ii) normalization $\Eoo(\R)=I$, and (iii) $\sigma$-additivity $\tr{\varrho \Eoo(\cup_j X_j)} = \sum_j\tr{\varrho \Eoo (X_j)}$ for all states $\varrho\in\trh$ and all sequences $(X_j)_j$ of pairwise disjoint Borel sets. It follows that for any state $\varrho$, the map $\varrho^\Eoo:\h B(\R)\to[0,1]$ defined via $\varrho^\Eoo(X) = \tr{\varrho \Eoo(X)}$ is a probability measure, and the number $\varrho^\Eoo(X)$ is interpreted as the probability that the measurement of $\Eoo$ gives an outcome from the set $X$, when the system is initially prepared in the state $\varrho$. In this article we are mainly interested in projection valued observables, that is, ones which satisfy $\Eoo(X)^2=\Eoo(X)$. By the spectral theorem \cite[X.4.11]{CFA90}, these are in one-to-one correspondence with selfadjoint operators on $\hi$.

The problem of phase retrieval can now be formulated as a problem of determining an unknown pure state from measurement outcome statistics. This is one instance of quantum tomography, a field which focuses on the problem of state reconstruction. In general, when the object to be determined is an arbitrary state, i.e., pure or mixed, then one needs to measure a collection $\h A$ of observables $\Eoo:\h B(\R)\to\lh$ which is {\em informationally complete} \cite{Prugovecki77} in the sense that for any two states $\varrho_1$ and $\varrho_2$, $\varrho^\Eoo_1 = \varrho^\Eoo_2$ for all $\Eoo\in\h A$ implies $\varrho_1=\varrho_2$. 
If one is only interested in determining pure states, then the following weaker notion is relevant (see, e.g., \cite{HeMaWo13,BuLa89}).

\begin{definition}\label{def:infocompleteness}
Let $\h A$ be a collection of observables $\Eoo:\h B(\R)\to\lh$. We say that $\h A$ is {\em informationally complete with respect to pure states} if  for any two pure states $\varrho_1$ and $\varrho_2$, $\varrho^\Eoo_1 = \varrho^\Eoo_2$ for all $\Eoo\in\h A$ implies $\varrho_1=\varrho_2$.
\end{definition}

Obviously, informational completeness implies informational completeness with respect to pure states. Moreover, it is clear that $\varrho^\Eoo_1 = \varrho^\Eoo_2$ if and only if $\sum_j c_j\tr{(\varrho_1-\varrho_2) \Eoo(X_j)} =0$ for all $(c_j)_j\subset \C$ and $(X_j)_j\subset \h B(\R)$. Therefore, it is the linear span of the operators $\Eoo(X)$ which is relevant for the purpose of determining the unknown state.

For a collection $\h A$ of observables $\Eoo:\h B(\R)\to\lh$, we denote by $\h R(\h A)$ the weak$^*$-closure of the complex linear span of the set $\{\Eoo(X) \mid \Eoo\in\h A, X\in \h B(\R)\}$. It follows that $\h R(\h A)$ is an {\em operator system}, that is,  a linear subspace of $\lh$ containing the identity $I$ of $\hh$ and satisfying $\h R(\h A)^*=\h R(\h A)$ (see \cite[p.~9]{CBMOA03}). We say that $\h R(\h A)$ is the operator system generated by $\h A$. We denote by $\h R(\h A)^\perp$ the annihilator of $\h R(\h A)$ in $\trh$, that is,
$$
\h R(\h A)^\perp = \{T\in\trh \mid \tr{TA}=0 \text{ for all }A\in\h R(\h A) \}.
$$
It is well known that a collection $\h A$ is informationally complete if and only if $\h R(\h A)^\perp =\{ 0\}$ \cite{Busch91}. 
Furthermore, a collection $\h A$ is informationally complete with respect to pure states if and only if every nonzero selfadjoint operator in $\h R(\h A)^\perp$ has rank $3$ or more \cite{CaHeScTo14}.
The following simple observation turns out to play a crucial role in our proofs.

\begin{proposition}
\label{prop:unitary_equivalence}
Let $\h A$ and $\h A^\prime$ be two collections of observables such that $\h R(\h A^\prime)=U\h R(\h A) U^\ast$ for some unitary operator $U$. 
Then $\h A^\prime$ is  informationally complete with respect to pure states if and only if  $\h A$ is such.
\end{proposition}

\begin{proof}
Suppose that $\h A$ is informationally complete with respect to pure states but $\h A^\prime$ is not. Then there exist two distinct pure states $\varrho_1$ and $\varrho_2$ such that $\tr{\varrho_1 \Eoo'(X)} = \tr{\varrho_2 \Eoo'(X)}$ for all $X\in\h B(\R)$ and $\Eoo'\in\aa'$, hence $\varrho_1-\varrho_2 \in \h R(\h A')^\perp$. We then have $U^*(\varrho_1-\varrho_2)U \in \h R(\h A)^\perp$, which implies $\tr{U^*\varrho_1 U\Eoo(X)} = \tr{U^*\varrho_2 U\Eoo(X)}$ for all $X\in\h B(\R)$ and $\Eoo\in\aa$. That is, the two distinct pure states $U^*\varrho_1 U$ and $U^*\varrho_2 U$ are not distinguished by $\h A$, which is a contradiction. Hence $\h A^\prime$ is informationally complete with respect to pure states. Interchanging the roles of $\h A$ and $\h A^\prime$ we have  the other implication.
\end{proof}

\section{Rotated quadrature observables}
We will now focus on the special case $\hi=L^2(\R)$. Physically this can be viewed as representing a single spinless particle confined to move in one spatial direction, or a single mode electromagnetic field. Let $Q$ and $P$ denote the standard position and momentum operators on $\hi$, so that $(Q\psi)(x) = x\psi(x)$ and $(P\psi)(x) =-i \psi'(x)$. For any $x=(q,p)^T\in\R^2$, define the corresponding Weyl operator 
\begin{equation}\label{eqn:Weyl_operator}
W(x) = e^{i\frac{qp}{2}}e^{-iqP}e^{ipQ} = e^{-iqP+ipQ}.
\end{equation}
The map $W:\R^2\to\lh$ is then an irreducible projective unitary representation of $\R^2$ which satisfies the composition rule
\begin{equation}\label{eqn:Weyl_commutation}
W(x)W(y) = e^{-\frac{i}{2} \{x,y\}} W(x+y), \qquad x,y\in\R^2
\end{equation}
where $\{(q,p)^T,(u,v)^T\} = qv-pu$ is the symplectic form on $\R^2$. In particular, the commutation relation $W(x)W(y) = e^{-i \{x,y\} } W(y)W(x)$ immediately follows. According to the Stone-von Neumann theorem \cite[(1.50)]{HAPS89}, any irreducible projective unitary representation of $\R^2$ satisfying 
 \eqref{eqn:Weyl_commutation} is unitarily equivalent to the standard one \eqref{eqn:Weyl_operator}. 

Now consider a symplectic matrix  $S\in SL(2,\R)$, i.e., one that satisfies $\{Sx,Sy\} = \{x,y\}$. Then clearly the map $x\mapsto W(Sx)$ satisfies \eqref{eqn:Weyl_commutation}, and therefore there exists  a unitary operator $U(S)$ such that $U(S)W(x)U(S)^* = W(Sx)$ for all $x\in\R^2$ (see \cite[Chapter 4.2]{HAPS89}).  
In particular, for any rotation  
$$
S_\theta = \left( \begin{array}{cc} \cos\theta & -\sin\theta \\ \sin\theta & \cos\theta  \end{array}\right)
$$
we obtain the corresponding unitary operator which, for the sake of clarity, we denote by $R(\theta)$. We can express this operator explicitly in terms of the orthonormal basis $\{h_n \mid n=0,1,\ldots\}$ of $\hi$ consisting of the  Hermite functions  
$$ 
h_n(x) = \frac{1}{\sqrt{2^n n! \sqrt{\pi}}} H_n(x) e^{-x^2/2}
$$
where
$$
H_n(x) = (-1)^n e^{x^2} \frac{\de^n}{\de x^n} e^{-x^2}
$$
is the $n^{\rm th}$ Hermite polynomial. Indeed, up to a phase factor we have 
$$
R(\theta) = \sum_{n=0}^\infty e^{i\theta n} \vert h_n\rangle\langle h_n \vert = e^{i\theta N}
$$
where $N=\sum_{n=0}^\infty n\vert h_n\rangle\langle h_n\vert$ is the number operator. We thus see that $R$ actually is a representation of the rotation group $SO(2)$ in $\hh$. Adopting the convention of  \cite{Jaming14} for the definition of the fractional Fourier transform, we have that $R(\theta)^* = R(-\theta) = F_\theta$, so that the adjoint of the rotation operator coincides with the fractional Fourier transform.

Now let $\Qo,\Po:\h B(\R)\to\lh$ be the position and momentum observables, namely, the projection valued measures associated with the operators $Q$ and $P$ by the spectral theorem. In particular, $\left[ \Qo(X)\psi\right](x)= \id_X(x) \psi(x)$, where $\id_X$ denotes the indicator function of the set $X$, and $\Po(X) = F^{-1} \Qo(X) F$ where $F=R(-\pi/2)$ is the unitary Fourier-Plancherel operator on $\hi$. For any  $\theta\in[0,2\pi)$ define the rotated quadrature observable 
$\Qo_\theta:\h B(\R)\to\lh$, 
$$
\Qo_\theta(X) = R(\theta)\Qo(X)R(\theta)^*.
$$ 
The corresponding rotated quadrature operators $Q_\theta$ are then the first moment operators of these observables, that is, $Q_\theta = \int x \Qo_\theta(\de x)$, and they may be expressed as $Q_\theta = Q\cos\theta + P\sin\theta$. 
For a system in a pure state $\varrho = \kb{\psi}{\psi}$, the measurement outcome probabilities related to the quadratures are given by
\begin{equation}\label{eq:quad-->Four}
\varrho^{\Qo_\theta}(X) = \langle \psi\vert \Qo_\theta(X) \psi\rangle = \langle R(\theta)^* \psi \vert \Qo(X)R(\theta)^*\psi\rangle = \int_X  \vert \left[ F_\theta \psi\right](x) \vert^2 \, \de x .
\end{equation}
This formula clarifies the aforementioned connection between quadrature observables and fractional Fourier transforms. 
 Indeed, it shows that the probability density associated to a measurement of the observable $\Qo_\theta$ performed on the pure state $\kb{\psi}{\psi}$ is just the intensity $|F_\theta \psi|^2$. Note that the probabilities $\{\varrho^{\Qo_\theta} \mid \theta \in [0,2\pi)\}$, or, equivalently, the intensities $\{|F_\theta \psi|^2 \mid \theta \in [0,2\pi)\}$ are also connected to the Wigner function $\ww(\vert \psi\rangle\langle \psi\vert)$ of the state \cite{CaGl69}. Namely, each density $\vert F_\theta \psi \vert^2$ coincides with the Radon transform of the Wigner function along the direction $\theta$ (see \cite{VoRi89}, and also \cite{ADT09} for a more precise review and statement of this fact).

The essential observation now is that the Fourier transform of the observable $\Qo_\theta$ is
\begin{equation}\label{eq:Weyl}
\int \e^{-iux} \Qo_\theta(\de x)=e^{-iuQ_\theta} = 
W(x) \quad \text{with} \quad x = (u\sin\theta,-u\cos\theta)^T ,
\end{equation}
where the integral is understood in the usual weak sense.
In other words, it corresponds to the restriction of the Weyl map $W$ to the one-dimensional subspace
\begin{equation}\label{eq:line}
L_\theta = \{(u\sin\theta,-u\cos\theta)^T\mid u\in\R\}
\end{equation}
of $\R^2$. With this observation, we can characterize the operator system generated by any set of rotated quadrature observables.

\begin{proposition}\label{prop:quadrature_operator_system}
Let $\h I\subset [0,2\pi)$. Then $\h R(\{  \Qo_{\theta}\mid \theta\in\h I\})$  is the weak$^*$-closure of the linear  span of 
$$
\bigcup_{\theta\in \h I}\{ W(x) \mid x\in L_\theta\}.
$$
\end{proposition}
\begin{proof}
By the bipolar theorem \cite[V.1.8]{CFA90}, it is enough to show that $T\in \h R(\{  \Qo_{\theta}\mid \theta\in\h I\})^\perp$ if and only if  $\tr{T W(x )}=0$ for all $x\in L_\theta$, $ \theta \in\h I$. Let $T\in \h R(\{  \Qo_{\theta}\mid \theta\in\h I\})^\perp$. Then for any $\theta\in\h I$ the complex measure $X\mapsto \tr{T\Qo_{\theta}(X)}$ is identically zero and by \eqref{eq:Weyl} we have $\tr{TW(x)} =0$ for all $x\in L_\theta$. 
Conversely, by \eqref{eq:Weyl} and the injectivity of the Fourier transform, the condition  $\tr{TW(x)}=0$ for all $x\in L_\theta$, $\theta\in\h I$, implies that $\tr{T\Qo_{\theta}(X)}$ for all $X\in\h B(\R)$ and $\theta\in\h I$ so that $T\in \h R(\{  \Qo_{\theta}\mid \theta\in\h I\})^\perp$. 
\end{proof}

\section{Main results}
We are now ready to prove the main results of this Letter. We begin by noting that, according to Proposition \ref{prop:quadrature_operator_system}, for all $\theta\in[0,\pi)$ the rotated quadratures $\Qo_\theta$ and $\Qo_{\theta+\pi}$ generate the same operator systems. Hence, it is always sufficient to consider quadratures with $\theta\in[0,\pi)$.  One of the consequences of Proposition \ref{prop:quadrature_operator_system} is that an arbitrary collection $\{  \Qo_{\theta}\mid \theta\in\h I\}$ of rotated quadratures is informationally complete if and only if $\h I$ is dense in $[0,\pi)$ (see, e.g., \cite{KiSc13}). 
Therefore no finite collection is able to distinguish between all  (i.e., pure or mixed) states. 
As the next proposition shows, even when restricting to pure states, there are certain collections which can be discarded.

\begin{proposition}\label{prop:rational_angles}
Let $\theta_1,\ldots,\theta_n\in[0,\pi)$  be such that $\theta_i-\theta_j\in\rational\,\pi$ for all $i,j=1,\ldots,n$. Then the collection $\{\Qo_{\theta_1},\ldots, \Qo_{\theta_n}\}$ is not informationally complete with respect to pure states.
\end{proposition}
\begin{proof}
Without loss of generality we may assume that $0=\theta_1 <\theta_2 <\ldots <\theta_n<\pi$. Indeed, if $\theta_1\neq 0$, then we can replace the observables $\Qo_{\theta_1},\ldots, \Qo_{\theta_n}$ with the rotated ones $R(\theta_1)^*\Qo_{\theta_1}R(\theta_1),\ldots,$ $R(\theta_1)^*\Qo_{\theta_n}R(\theta_1)$ since the unitary transformation does not affect the property of informational completeness with respect to pure states by Proposition \ref{prop:unitary_equivalence}. Since  $R(\theta_1)^*\Qo_{\theta_j}R(\theta_1) = \Qo_{\theta_j-\theta_1}$, by denoting $\theta'_j=\theta_j-\theta_1$, we see that the observables $\Qo_{\theta'_1},\ldots, \Qo_{\theta'_n}$ satisfy $\theta'_1=0$. 

For each $j=2,\ldots,n$ there exist $q_j,p_j\in\nat$ such that $\theta_j = \frac{q_j}{p_j}\pi$. By setting $k=2\cdot p_2\cdots p_n$ we have that $k\,\theta_j=0\, ({\rm mod}\, 2\pi)$ for all $j=1,\ldots, n$. In particular, $R(\theta_j)^*h_k =\e^{-ik\theta_j}h_k = h_k$, so that by defining $\psi_\pm = \frac{1}{\sqrt{2}}(h_0  \pm ih_k)$ we have $R(\theta_j)^*\psi_\pm  = \psi_\pm$. 
The pure states $\vert\psi_+\rangle \langle \psi_+ \vert$ and $\vert\psi_-\rangle \langle \psi_-\vert$ are distinct, but 
\begin{eqnarray*}
 \langle \psi_\pm \vert \Qo_{\theta_j}(X) \psi_\pm \rangle &=& \langle \psi_\pm \vert \Qo(X) \psi_\pm \rangle \\
  &=& \int_X \frac{1}{2} \left(\vert h_0(x)\vert^2 \pm i \overline{h_0(x)} h_k(x) \mp i\overline{h_k(x)} h_0(x) + \vert h_k(x) \vert^2 \right) \de x \\
& =& \int_X \frac{1}{2} \left(\vert h_0(x)\vert^2 + \vert h_k(x) \vert^2 \right) \de x
\end{eqnarray*}
since the Hermite functions are real valued. Hence, the pure states $\vert\psi_+\rangle \langle \psi_+ \vert$ and $\vert\psi_-\rangle \langle \psi_-\vert$ cannot be distinguished and therefore  $\{\Qo_{\theta_1},\ldots, \Qo_{\theta_n}\}$ is not informationally complete with respect to pure states.
\end{proof}

When $n=2$ and $\theta_1-\theta_2\in\rational\,\pi$, the fact that the collection of two observables $\{\Qo_{\theta_1},\Qo_{\theta_2}\}$ is not informationally complete with respect to pure states was already observed in \cite[Remark 5.7]{Jaming14}.
It is worth noting that the assumption of finiteness for the collection of rotated quadratures is crucial in Proposition \ref{prop:rational_angles} above. Indeed, by going to infinitely many quadratures it is easy to give examples where the corresponding statement is false. The most simple example is given by the collection $\{ \Qo_{\theta} \mid \theta \in \rational\, \pi \cap [0,\pi) \}$. Since $\rational\, \pi \cap [0,\pi) $ is dense in $[0,\pi)$, this collection of observables is even informationally complete. 

We now come to our main result.

\begin{proposition}\label{prop:three_quadratures}
Let $\theta_1,\theta_2,\theta_3\in[0,\pi)$. 
The collection of observables $\{\Qo_{\theta_1},\Qo_{\theta_2}, \Qo_{\theta_3}\}$ is not  informationally complete with respect to pure states.
\end{proposition}
\begin{proof}
We may assume that $0\leq\theta_1 <\theta_2 <\theta_3 <\pi$. We will show that there exists a unitary operator $U$ such that $U\h R(\Qo_{\theta_1},\Qo_{\theta_2}, \Qo_{\theta_3})U^* =\h R(\Qo_{0},\Qo_{\pi/4}, \Qo_{\pi/2}) $. 
It follows from Proposition \ref{prop:rational_angles} that the collection $\{\Qo_{0},\Qo_{\pi/4}, \Qo_{\pi/2}\}$ is not informationally complete with respect to pure states, and Proposition \ref{prop:unitary_equivalence} implies the same for $\{\Qo_{\theta_1},\Qo_{\theta_2}, \Qo_{\theta_3}\}$.

The proof of the unitary equivalence goes as follows. By Proposition \ref{prop:quadrature_operator_system} the operator system $\h R(\Qo_{\theta_1},\Qo_{\theta_2},\Qo_{\theta_3})$ is the weak$^*$-closure of the linear span of the Weyl operators $ W(x)$  with $x\in \bigcup_{j=1}^3 L_{\theta_j}$, where $L_{\theta_j}$ are the lines defined in \eqref{eq:line}. 
Given any symplectic matrix $S\in SL(2,\R)$, the Stone-von Neumann theorem \cite[(1.50)]{HAPS89} shows the existence of a unitary operator $U(S)$ such that $U(S)W(x)U(S)^* = W(Sx)$ for all $x\in\R^2$. 
Therefore,
$$
U(S)\{ W(x) \mid x\in L_{\theta_1}\cup L_{\theta_2}\cup L_{\theta_3}\} U(S)^* = \{ W(x) \mid x\in L_{\theta'_1}\cup L_{\theta'_2}\cup L_{\theta'_3}\} \quad \text{with } \quad L_{\theta'_j}=SL_{\theta_j}.
$$
By Proposition \ref{prop:quadrature_operator_system}, we then have $U(S)\h R(\Qo_{\theta_1},\Qo_{\theta_2}, \Qo_{\theta_3})U(S)^* = \h R(\Qo_{\theta'_1},\Qo_{\theta'_2}, \Qo_{\theta'_3})$. Thus, we only need to find a symplectic matrix $S$ such that $SL_{\theta_1} = L_0$, $SL_{\theta_2}= L_{\pi/4}$ and $SL_{\theta_3} = L_{\pi/2}$. 
It is easy to check that the matrix
$$
S = \frac{1}{\sqrt{\sin(\theta_3-\theta_1)}}
\left(
\begin{array}{cc}
\sqrt{\frac{\sin(\theta_3-\theta_2)}{\sin(\theta_2-\theta_1)}}
\cos\theta_1 & \sqrt{\frac{\sin(\theta_3-\theta_2)}{\sin(\theta_2-\theta_1)}}
\sin\theta_1 \\
\sqrt{\frac{\sin(\theta_2-\theta_1)}{\sin(\theta_3-\theta_2)}} \cos\theta_3 & \sqrt{\frac{\sin(\theta_2-\theta_1)}{\sin(\theta_3-\theta_2)}} \sin\theta_3
\end{array}
\right)
$$
has the required property. Indeed, we have
\begin{align*}
S \left(\begin{array}{c} \sin\theta_1 \\ -\cos\theta_1 \end{array}\right) & = \sqrt{\frac{\sin(\theta_3-\theta_1)\sin(\theta_2-\theta_1)}{\sin(\theta_3-\theta_2)}} \left(\begin{array}{c} \sin 0 \\ -\cos 0 \end{array}\right) \\
S \left(\begin{array}{c} \sin\theta_2 \\ -\cos\theta_2 \end{array}\right) & = \sqrt{\frac{2\sin(\theta_3-\theta_2)\sin(\theta_2-\theta_1)}{\sin(\theta_3-\theta_1)}} \left(\begin{array}{c} \sin\frac{\pi}{4} \\ -\cos\frac{\pi}{4} \end{array}\right) \\
S \left(\begin{array}{c} \sin\theta_3 \\ -\cos\theta_3 \end{array}\right) & = \sqrt{\frac{\sin(\theta_3-\theta_2)\sin(\theta_3-\theta_1)}{\sin(\theta_2-\theta_1)}} \left(\begin{array}{c} \sin\frac{\pi}{2} \\ -\cos\frac{\pi}{2} \end{array}\right) .
\end{align*}
\end{proof}

Note that if the result of Proposition \ref{prop:three_quadratures} is regarded from the point of view of the Wigner transform $\ww[\psi] = \ww(\vert \psi\rangle\langle \psi\vert)$ of the signal $\psi\in L^2(\R)$, then it implies that knowing the Radon transform of $\ww[\psi]$ along only three directions $\theta_1, \theta_2$ and $\theta_3$ is not enough to reconstruct $\psi$ when $\psi$ is arbitrary.

\section{Discussion}

Since our proof of the main result regarding the insufficiency of three rotated quadratures, or equivalently fractional Fourier transforms, came down to finding a single suitable symplectic matrix, one might hope to use the same approach also for four quadratures. Indeed, four quadratures $\{\Qo_{\theta_1},\Qo_{\theta_2},\Qo_{\theta_3},\Qo_{\theta_4}\}$ are associated to the corresponding lines $L_{\theta_j}$, $j=1,\ldots, 4$, by Proposition \ref{prop:quadrature_operator_system}, and we would then try to find a symplectic matrix $S$ which maps the lines in such a way that $SL_{\theta_j} = L_{\theta_j'}$ where now $\theta_i'-\theta_j'\in\rational\, \pi$ for all $i,j=1,\ldots,4$. However, a moment's thought reveals the fact that this is not always possible. 

Indeed, suppose that we have $\theta_1=0$, $\theta_2=\pi/4$, $ \theta_3 = \pi/2$ and $\theta_4\in (\pi/2,\pi)$ is such that $\cot\theta_4$ is a transcendental number. By applying an extra rotation if necessary, we may assume that $\theta_1'=0$. Now the conditions $SL_0  = L_0$ and $\det S=1$ imply that 
$$
S =
\left(
\begin{array}{cc}
a & 0\\
b & a^{-1}
\end{array}
\right)
$$
so that by looking at how the slope $k_\theta =-\cot{\theta}$ of the line $L_\theta$ changes, we obtain the equations
$$
\cot\theta_j' = \frac{a^{-1} \cos\theta_j -b\sin\theta_j}{a\sin\theta_j}
$$
for $j=2,3,4$. In particular, it follows from the cases $j=2$ and $j=3$ that
$$
a^2 = \frac{1}{\cot\theta'_2 - \cot\theta'_3}
\quad \text{and} \quad
b = -a\cot\theta'_3 
$$
so that both $a$ and $b$ must be algebraic numbers \cite{Swift22}. But the remaining case $j=4$ now gives 
$$
\cot\theta_4= a^2\cot\theta_4'+ab
$$
which is a contradiction since the right-hand side is an algebraic number, but the left-hand side is transcendental by assumption.

Coming back to the quadratures $\{\Qo_{\theta_1},\Qo_{\theta_2},\Qo_{\theta_3},\Qo_{\theta_4}\}$, we may now conclude that our method fails to prove or disprove their informational completeness with respect to pure states when the four angles $\theta_1,\theta_2,\theta_3,\theta_4$ are chosen as above. Indeed, picking an arbitrary $S\in SL(2,\R)$, we still obtain the equality $U(S)\h R(\Qo_0,\Qo_{\pi/4},\Qo_{\pi/2},\Qo_{\theta_4}) U(S)^*= \h R(\Qo_{\theta'_1},\Qo_{\theta'_2},\Qo_{\theta'_3},\Qo_{\theta'_4})$. However, by the previous discussion it is not possible to make all the differences $\theta'_i-\theta'_j$ rational multiples of $\pi$, thus Proposition \ref{prop:rational_angles} can no longer be invoked to conclude that the collection $\{\Qo_{\theta'_1},\Qo_{\theta'_2},\Qo_{\theta'_3},\Qo_{\theta'_4}\}$ is not informationally complete with respect to pure states. Therefore, the sufficiency of the four quadratures $\{\Qo_0,\Qo_{\pi/4},\Qo_{\pi/2},\Qo_{\theta_4}\}$ to determine all pure states still remains an open question when $\cot\theta_4$ is a transcendental number.

\section*{Acknowledgments}
TH acknowledges financial support from the Academy of Finland (grant no 138135). JS and AT acknowledge financial support from the Italian Ministry of Education, University and Research (FIRB project RBFR10COAQ).

\end{document}